\documentclass[letterpaper]{article}
\usepackage{authblk}

\usepackage{uai2019}

\usepackage[hyphens]{url}
\usepackage{hyperref}
\usepackage[margin=1in]{geometry}
\usepackage{esdiff}
\usepackage{graphicx}
\usepackage{hyperref}
\usepackage{booktabs} 
\usepackage[round]{natbib}
\setcitestyle{authoryear,round,citesep={;},aysep={,},yysep={;}}
\usepackage {hyperref}
\usepackage{todonotes}
\usepackage{caption}
\usepackage{subcaption}
\usepackage{comment}
\usepackage{textcomp}

\usepackage{amsmath,amsthm,amssymb}
\newtheorem{theorem}{Theorem}
\newtheorem{definition}{Definition}
\newtheorem{prop}{Proposition}
\usepackage{appendix}
\def\E{{\textrm{E}}\,}

\def\Cov{{\textrm{Cov}}\,}

\usepackage{bbm, dsfont}

\usepackage{times}

\title{The Elliptical Processes: \\a Family of Fat-tailed Stochastic Processes}
\author[1,2]{\bf{Maria B\.ankestad}}
\author[3]{\bf{Jens Sj\"olund}}
\author[1]{\bf{Jalil Taghia}}
\author[1]{\bf{Thomas Sch\"on}}
\affil[1]{Department of Information Technology, Uppsala Univerity}
\affil[2]{RISE Research Institutes of Sweden AB}
\affil[3]{Research and Physics, Elekta}

\begin{document}

\maketitle
\begin{abstract}
  We present the elliptical processes---a family of non-parametric probabilistic models that subsumes the Gaussian process and the \mbox{Student-t} process. This generalization includes a range of new fat-tailed behaviors yet retains computational tractability. We base the elliptical processes on a representation of elliptical distributions as a continuous mixture of Gaussian distributions and derive closed-form expressions for the marginal and conditional distributions. We perform numerical experiments on robust regression using an elliptical process defined by a piecewise constant mixing distribution, and show advantages compared with a Gaussian process. The elliptical processes may become a replacement for Gaussian processes in several settings, including when the likelihood is not Gaussian or when accurate tail modeling is critical.

\end{abstract}

\section{INTRODUCTION}
Stochastic processes can be seen as probability distributions over functions. As such, they provide a starting point for Bayesian non-parametric regression. The most prominent example is the Gaussian process \citep{Rasmussen2006}, which is one of the most popular methods for nonlinear regression due to its flexibility, interpretability and probabilistic nature.

When a Gaussian process prior is combined with a Gaussian likelihood, the resulting marginal and conditional distributions have simple and exact expressions. But, if the underlying data is \emph{not} Gaussian, the conditional distribution can be seriously misleading. We describe a new family of elliptical processes that can adapt to such situations. 
\begin{figure}[t]
    \vspace{.3in}
    \centerline{\includegraphics[width=0.99\columnwidth]{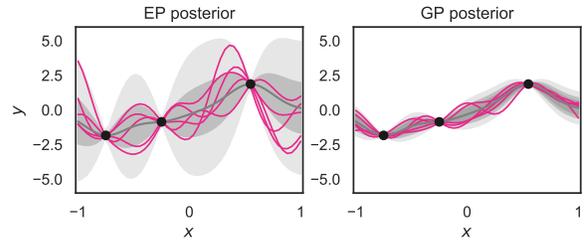}}
    \vspace{.1in}
    \caption{Posterior distributions of an elliptical process and a Gaussian process with the same kernel hyperparameters.}
    \label{fig:plot_of_elliptical_process}
\end{figure}
The elliptical processes subsumes the Gaussian process and the \mbox{Student-t} process \citep{Shah2014}. It is based on the elliptical distribution---a broad family of distributions that includes the Gaussian and \mbox{Student-t} distributions---which is attractive because it can describe fat-tailed distributions while retaining most of the Gaussian distribution's computational tractability \citep{Fang1990}. For these reasons, the elliptical distribution is widely used in finance, especially in portfolio theory \citep{Owen1983,Gupta2013}.

Every elliptical process corresponds to a mixing distribution. Interestingly, the converse also holds (subject to a few mild conditions)---any distribution over the non-negative real numbers gives rise to a unique elliptical process. To make this large class of stochastic processes manageable, we perform a detailed analysis of the elliptical process that results from a piecewise constant mixing distribution.

The added flexibility of elliptical processes could benefit a spectrum of applications. On one end is robust regression, where the conditional mean is the primary focus and outliers are considered a nuisance that should be de-emphasized. We exemplify this using the elliptical process on both synthetic and real-world data sets. At the other end of the spectrum are cases where the tail behavior---extreme values in particular---is the primary focus, for instance when modeling the probabilities of rare, but possibly extreme, events such as natural disasters, stock market crashes and global pandemics \citep{King2001, Ghil2011}. 

\section{RELATED WORK}

Attempts at making Gaussian process regression more robust can be broadly categorized into those modifying the likelihood and those modifying the stochastic process prior. In both cases, however, a natural first step is to replace the underlying Gaussian model with a Student-t model. Replacing the Gaussian likelihood with a Student-t likelihood, makes the regression more robust against outliers, but requires approximate inference \citep{Neal1997,Jylanki2011}.

The Student-t process can be defined as a scale-mixture of Gaussian processes \citep{OHagan1999}. As pointed out by \citet{Rasmussen2006}, this works well in the noise-free setting, but inclusion of independent noise comes at the cost of analytic tractability. \citet{Shah2014}, on the other hand, preserved analytic tractability by instead including noise in the covariance kernel (making the noise uncorrelated but not independent) and demonstrated the empirical effectiveness of this approach.


\section{ELLIPTICAL DISTRIBUTIONS}
Our new elliptical process is based on elliptical distributions, which are interesting because they include Gaussian distributions as well as more fat-tailed distributions. In this section, we review the relevant background on elliptical distributions.

Let $\mathbf{y}\in \mathbb{R}^n$ be a random variable that admits a probability density with respect to the Lebesgue measure. \citet{kelker1970distribution} showed that if $\mathbf{y}$ follows the elliptical distribution its density can be expressed as
    \begin{align}\label{eq:elliptical_density}
        p_{\boldsymbol{\theta}}(u)
        &=c_{n, \boldsymbol{\theta}}|\boldsymbol{\Sigma}|^{-1/2}g_{n, \boldsymbol{\theta}}(u),
    \end{align}
where $u = (\mathbf{y}-\boldsymbol{\mu})^T\boldsymbol{\Sigma}^{-1}(\mathbf{y}-\boldsymbol{\mu})$ is the squared Mahalanobis distance, $\boldsymbol{\mu}$ is the location vector, $\boldsymbol{\Sigma}$ is the scale matrix, and $c_{n, \boldsymbol{\theta}}$ is a normalization constant. The density generator $g_{n, \boldsymbol{\theta}}(u)$ is a non-negative function with finite integral, parameterized by $\boldsymbol{\theta}$, which determines the shape of the distribution. We recover the Gaussian distribution if $g_n(u) = \exp \left\{ -\frac{u}{2}\right \}$.


\begin{figure}[ht]
    \centerline{\includegraphics[width=0.9\columnwidth]{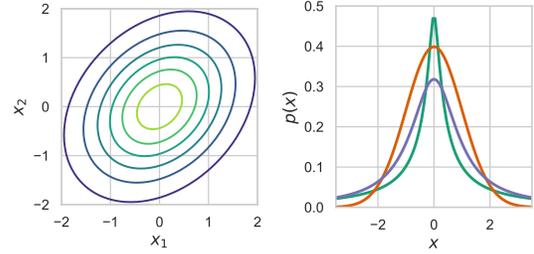}}
    \vspace{.1in}
    \caption{\textbf{Left:} A contour plot of an elliptical two dimensional, correlated, distribution with zero mean. The name derives from its elliptical level sets. \textbf{Right:} Three examples of one dimensional elliptical distributions with zero mean. A common feature of the elliptical distributions is that they are symmetric about the mean, $\E[\mathbf{X}] =\boldsymbol{\mu}$.}
    \label{fig:elliptic_distribution}
    \vspace{.1in}
\end{figure}


The Gaussian distribution---the basic building block of Gaussian processes---has several attractive properties that we wish the elliptical process to inherit:
\begin{itemize}
    \item Closure under marginalization: if we marginalize over a variable in a multivariate Gaussian distribution, the result is still a Gaussian distribution.
    \item Closure under conditioning: the conditional distribution is also a Gaussian distribution.
    \item Straightforward sampling. 
\end{itemize}
This leads us to consider the family of \emph{consistent} elliptical distributions. Following \citet{kano1994consistency}, we say that a family of elliptical distributions $\{p_\theta (u(\mathbf{y}_{n}))\,|\, n\in\mathbb{N}\}$ is consistent if and only if
\begin{equation}
    \int_{-\infty}^\infty p_{\boldsymbol{\theta}}\left( u(\mathbf{y}_{n+1})\right) d{y}_{n+1} = p_{\boldsymbol{\theta}}\left( u(\mathbf{y}_{n}) \right).
\end{equation}
In other words, a consistent elliptical distribution is closed under marginalization. 

Far from all elliptical distributions are consistent, but the complete characterization of those that are is given by the following theorem \citep{kano1994consistency}.
\begin{theorem}\label{theorem:consistent}
    An elliptical distribution is consistent if and only if it originates from the integral
    \begin{equation}\label{eq:consistence_folrmula}
        p_{\boldsymbol{\theta}}(u)  =|\boldsymbol{\Sigma}|^{-\frac{1}{2}} \int_0^{\infty} \left ( \frac{\xi}{2\pi}\right)^{\frac{n}{2}}e^{\frac{-u\xi}{2}}p_{\boldsymbol{\theta}}(\xi)d\xi,
    \end{equation}
    where $\xi$ is a mixing variable with the corresponding, strictly positive, mixing distribution $p_{\boldsymbol{\theta}}(\xi)$, that is independent of $n$ and satisfies $p_{\boldsymbol{\theta}}(0) = 0$.
\end{theorem}
$p_{\boldsymbol{\theta}}(u) $ is a scale mixture of Gaussian distributions, with the mixing variable $\xi \sim p_{\boldsymbol{\theta}}(\xi)$. Note that any choice of mixing distribution that fulfills the conditions in Theorem\,\ref{theorem:consistent} can be used to define a consistent elliptical process. We recover the Gaussian distribution if the mixing distribution is a Dirac delta function and the Student-t distribution if it is a scaled chi-square distribution. 

\section{ELLIPTICAL PROCESSES}
Similar to how Gaussian Processes are defined, we define an elliptical process as follows:
\begin{definition}\label{def:ellipticalProcess}
    An elliptical process ($\mathcal{EP}$) is a collection of random variables such that every finite subset has a consistent elliptical distribution, where the scale matrix is given by a covariance kernel. 
\end{definition}
This means that an elliptical process is specified by a mean function ${\mu}(\mathbf{x})$, scale matrix (kernel) $k(\mathbf{x}, \mathbf{x})$ and mixing distribution $p_{\boldsymbol{\theta}}(\xi)$.

Formally a stochastic process $\{X_t:t\in T\}$ on a probability space $(\Omega, \mathcal{F}, P)$ is a family of random maps $X_t:\Omega\to S_t$, $t\in T$, for measurable spaces $(S_t, \mathcal{S}_t)$, $t\in T$ \citep{Bhattacharya2007}. We focus on the setting where $S=\mathbb{R}$ and the index set $T$ is a subset of $\mathbb{R}^n$, in particular the halfline $[0, \infty)$.

That we may construct the elliptical process from the family of finite-dimensional, consistent, elliptical distributions follows from Kolmogorov's extension theorem, which is easy to check because of the restriction to $S=\mathbb{R}$ (which is a Polish space) and Kano's characterization above. 

Since the elliptical process is built upon consistent elliptical distributions it's closed under marginalization. The marginal mean $\boldsymbol{\mu}$ is the same as the mean for the Gaussian distribution. The covariance is defined by: 
\begin{equation}
    \Cov[\mathbf{Y}] = \E\left[\xi^{-1}\right]\boldsymbol{\Sigma},
\end{equation}
where $\boldsymbol{\Sigma}$ is the covariance for a Gaussian distribution and $\xi$ is the mixture variable in Equation \ref{eq:consistence_folrmula}. 

\subsection{Conditional Distribution}
To use the elliptical process for regression, we need the conditional mean and covariance of the corresponding elliptical distribution, which we derive next.

We partition the data as $\mathbf{y}=[\mathbf{y}_1, \mathbf{y}_2]$, where $\mathbf{y}_1$ are the $n_1$ observed data points, $\mathbf{y}_2$ are the $n_2$ data points to predict, and $n_1 + n_2 = n$. We have the following result:

\begin{prop}\label{prop:conditional_distribution}
    If the data $\mathbf{y}=[\mathbf{y}_1, \mathbf{y}_2]$ originate from the consistent elliptical distribution in equation \eqref{eq:consistence_folrmula}, the conditional distribution originates from the distribution
    \begin{equation}
        \resizebox{.96\hsize}{!}{$p_{\mathbf{y_2}|u_1}(\mathbf{y}_2)  =
        \frac{c_{n_1,\boldsymbol{\theta}}}{\left|\boldsymbol{\Sigma}_{22|1}\right|^{\frac{1}{2}}(2\pi)^{\frac{n_2}{2}}} \int_0^{\infty}\xi^{\frac{n}{2}}e^{-(u_{2|1} + u_1)\frac{\xi}{2}}p_{\boldsymbol{\theta}}(\xi)d\xi$}
    \end{equation}
    with the conditional mean
    \begin{equation}
        \E[\boldsymbol{y_1}|\boldsymbol{y}_2] = \boldsymbol{\mu}_{2|1} 
    \end{equation}
    and the conditional covariance
    \begin{equation}
        \Cov[\boldsymbol{Y}_1|\boldsymbol{Y}_2 = \boldsymbol{y}_2] = \E [\hat{\xi}^{-1}]\boldsymbol{\Sigma}_{22|1}, \ \ \hat{\xi}\sim \xi | \mathbf{y}_1.
    \end{equation}
    where $u_1 = (\mathbf{y}_1-\boldsymbol{\mu}_1)^T\boldsymbol{\Sigma}_{11}^{-1}(\mathbf{y}_1-\boldsymbol{\mu}_1)$, $u_{2|1} = (\mathbf{y}_2-\boldsymbol{\mu}_{2|1})^T\Sigma_{22|1}^{-1}(\mathbf{y}_2-\boldsymbol{\mu}_{2|1})$, and $c_{n_1,\boldsymbol{\theta}}$ is a normalization constant.
\end{prop}

\begin{proof}
    See Appendix \ref{appendix:conditional_dist}.
\end{proof}
The conditional distribution is guaranteed to be a consistent elliptical distribution, but not necessarily the same as the original one. (Recall that consistency only concerns the marginal distribution). Since the conditional distribution depend on $u_1$ and $n_1$, the shape depends on the training samples.

The conditional scale matrix $\boldsymbol{\Sigma}_{22|1}$ and the conditional mean vector $\boldsymbol{\mu}_{2|1}$ are the same as the mean and the covariance matrix for a Gaussian distribution. To get the covariance for our elliptical distribution we multiply the scale matrix with a constant that depends on the training data $\mathbf{y}_1$, namely $\E\left[\hat{\xi}^{-1}\right]$. This is the same behavior as for the Student-t process \citep{Shah2014}, but the constant differs.

Equipped with the conditional mean and covariance for the elliptical distribution we can make predictions on unseen data points. 

\subsection{Training}
Gaussian processes are widely used for regression. They are convenient to work with because the combination of a Gaussian process prior with a Gaussian likelihood remains a Gaussian process. Unfortunately, this closure property does \emph{not} hold for elliptical distributions in general. To add noise, we either have to simplify the model or use approximate inference. 

In this paper we use the same procedure as in \citet{Shah2014}, namely to add noise to the kernel, 

\begin{equation}\label{eq:delta_noise_kernel}
    \boldsymbol{\Sigma} = \mathbf{K} + \epsilon\,\mathbf{I}.
\end{equation}

Here, $\mathbf{K}_{ij} = {k}(\mathbf{x}_i, \mathbf{x}_j)$ is the kernel matrix and $\epsilon$ is the noise: a constant added to the diagonal elements of $\boldsymbol{\Sigma}$. According to the law of total expectation, this makes the noise not independent, but uncorrelated, with the latent function (see Appendix \ref{appendix:added_noise}).

\begin{figure*}[t!]
    \vskip 0.01in
    \centerline{\includegraphics[width=1.5\columnwidth]{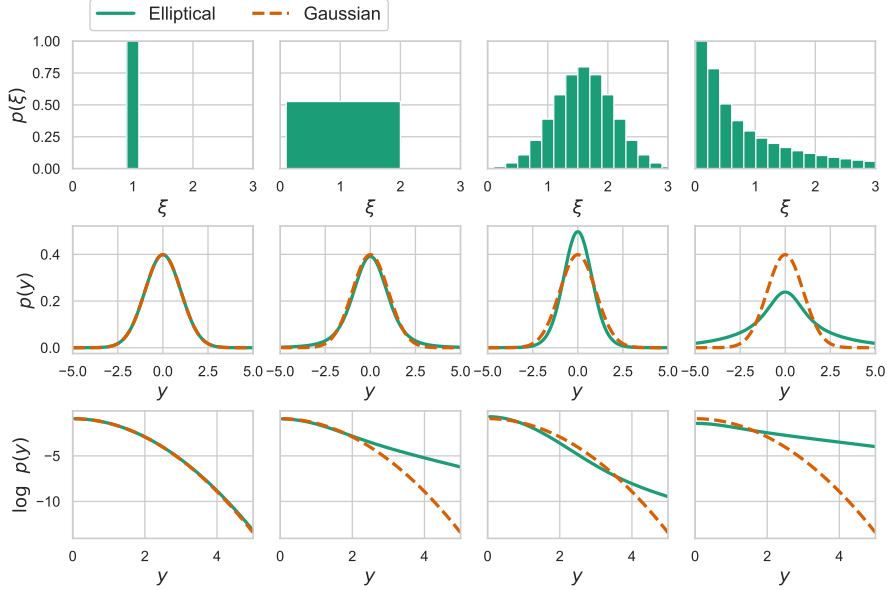}}
    \vskip 0.1in
    \caption{Different mixing distributions $p_{\boldsymbol{\theta}}(\xi)$ (\textbf{top}) correspond to different elliptical distributions $p_{\boldsymbol{\theta}}(y)$ (\textbf{middle}) with fatter tails than the Gaussian distribution, which is evident in log-scale (\textbf{bottom}).} 
    \label{fig:examples_elliptical}
\end{figure*}

For a consistent elliptical distribution $\Cov(\mathbf{X}) = \E\left [\xi^{-1}\right] \boldsymbol{\Sigma}$, so both the noise and the kernel matrix are scaled by $\E\left [ \xi^{-1}\right]$. For the Gaussian distribution, $\xi$ follows a Dirac delta function, $\delta(\xi-1)$, so $\E\left [ \xi^{-1}\right] =1$ and $\Cov(\mathbf{X})= \boldsymbol{\Sigma}$ (leading to independent noise).

Thanks to the simplified noise model we can train the model by maximizing the (exact) marginal likelihood.

\section{EXAMPLE: PIECEWISE CONSTANT MIXING DISTRIBUTION}\label{sec:PiecewiseConstant}

Based on the observation that the integral in Theorem\,\ref{theorem:consistent} is effectively a Gamma function when the mixing distribution $p_{\boldsymbol{\theta}}(\xi)$ is constant, we construct a mixing distribution by stacking piecewise constant functions next to each other, as illustrated in Figure \ref{fig:elliptical_distribution}. 

\begin{figure}[ht]
    \vskip 0.2in
    \centerline{\includegraphics[width=0.8\columnwidth]{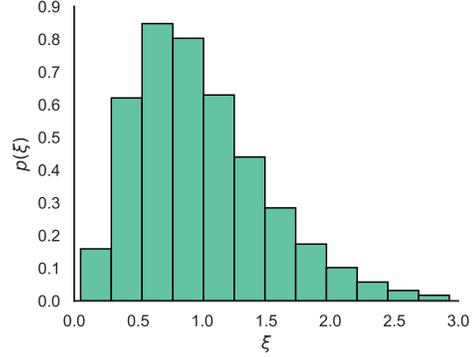}}
    \vspace{.1in}
    \caption{A piecewise constant mixing distribution $p_{\boldsymbol{\theta}}(\xi)$ with 12 pieces.}
    \label{fig:elliptical_distribution}
    \vspace{.15in}
\end{figure}

Each piece has constant width and variable height. Provided that $p_{\boldsymbol{\theta}}(0) = 0$ we may approximate any desired shape of $p_{\boldsymbol{\theta}}(\xi)$. The mixing distribution with $M$ pieces of heights $h_1, \hdots, h_M$, width $\Delta$ and starting position $\ell_0$ is given by

\begin{equation}\label{eq:elliptical_dist}
    p_{\boldsymbol{\theta}}(\xi) = \frac{1}{\Delta\sum_k h_k}\sum_{k=1}^M h_k \mathbbm{1}_{(\ell_{k-1}, \ell_k)}(\xi),
\end{equation}

where $\ell_k = \ell_0 + k\Delta$ and, $\mathbbm{1}_{(a, b)}(\xi) = 1$ when $a<\xi < b$ and zero otherwise. To modify the distribution, we change the heights $h_1, \hdots, h_M $.

The corresponding consistent elliptical distribution is (see Appendix \ref{appendix:deriving_sq_box_dis})
\begin{equation}\label{eq:marginalLikelihoodPiecewise}
        p_{\boldsymbol{\theta}}(u) = \frac{2u^{-(\frac{n}{2}+1)}}{ |\Sigma|^{\frac{1}{2}}\pi^{\frac{n}{2}}\Delta}
        \Phi_{M, \mathbf{h}, \boldsymbol{\ell}}\left (\frac{n}{2} + 1, u\right ),
\end{equation}
where
\begin{equation}
        \Phi_{M, \mathbf{h}, \boldsymbol{\ell}}(s, u) = \sum_{k=1}^M  h_k\Gamma\left (s,\frac{\ell_{k-1}u}{2}, \frac{\ell_k u}{2} \right)
\end{equation}
and $\Gamma(\cdot,\cdot,\cdot)$ is the incomplete Gamma function. 

Figure \ref{fig:examples_elliptical} shows examples of piecewise constant mixing distribution $p_{\boldsymbol{\theta}}(\xi)$, with corresponding elliptical distributions $p_{\boldsymbol{\theta}}(y)$ (recall that $u = (\mathbf{y}-\boldsymbol{\mu})^T\boldsymbol{\Sigma}^{-1}(\mathbf{y}-\boldsymbol{\mu})$ ) and the logarithm of their tails.

We train the elliptical process by minimizing the negative logarithm of the marginal likelihood in Equation \eqref{eq:marginalLikelihoodPiecewise} with the addition of a smoothness penalty on the mixing distribution parameters $h_1,\hdots, h_M$ (see Appendix \ref{appendix:implementation}).
Those are trained jointly with the kernel parameters $\boldsymbol{\lambda}$ and the noise level $\epsilon$. 

The elliptical process may be used for classification and generative modelling, but in this paper, we mainly focus on regression, using the following expressions for the conditional mean and covariance (see Appendix \ref{appendix:deriving_sq_box_dis}):
\begin{subequations}
        \begin{align}
        \E[\mathbf{y}_2|\mathbf{y}_1] &= \boldsymbol{\mu}_{2|1}, \\
        \Cov(\mathbf{y}_2| \mathbf{y}_1)&= \frac{u_1}{2}\frac{ \Phi_{M, \mathbf{h}, \boldsymbol{\ell}}\left (\frac{n_1}{2}, u_1\right )}{
       \Phi_{M, \mathbf{h}, \boldsymbol{\ell}}\left (\frac{n_1}{2} +1, u_1\right )} \boldsymbol{\Sigma}_{22|1},\\
        \intertext{where}
        u_1 &= (\mathbf{y}_1-\boldsymbol{\mu}_1)^T\boldsymbol{\Sigma}_{11}^{-1}(\mathbf{y}_1-\boldsymbol{\mu}_1), \\
        \boldsymbol{\mu}_{2|1} &= \boldsymbol{\mu}_2 +\boldsymbol{\Sigma}_{21}\boldsymbol{\Sigma}_{11}^{-1}(\mathbf{y}_1-\boldsymbol{\mu}_{1}),\\
        \boldsymbol{\Sigma}_{22|1} &= \boldsymbol{\Sigma}_{22}- \boldsymbol{\Sigma}_{21}\boldsymbol{\Sigma}_{11}^{-1}\boldsymbol{\Sigma}_{12}.
        \end{align}
\end{subequations}



\section{EXPERIMENTS}
We examined the elliptical process with piecewise constant mixing distribution, described in Section \ref{sec:PiecewiseConstant}, on synthetic and real-world data.  

\subsection{Implementation}
In the experiments we set the number of pieces to $M=10$, the box lengths to $\Delta = 0.2$, and the starting position to $\ell_0 = 0.01$. 
This means that there were 10 parameters of the mixing distribution to optimize, $h_1, \hdots, h_{10}$, in addition to the kernel hyperparameters and the noise variance $\epsilon$.

In all experiments we used a squared exponential kernel together with a Kronecker delta function, as in equation \eqref{eq:delta_noise_kernel}. We trained our model by minimizing the negative logarithm of the marginal likelihood in equation \eqref{eq:marginalLikelihoodPiecewise} using the Adam optimizer \citep{kingma2014adam}. See Appendix \ref{appendix:implementation} for further details on the implementation. The code for the experiments will be published on GitHub if the paper is accepted.

\subsection{Investigating the Tail Behavior}
\begin{figure}[t!]
\centering
\begin{subfigure}{\columnwidth}
\caption{Scaled chi-square distribution}\label{fig:qq_plot_1}
  \includegraphics[width=0.9\linewidth]{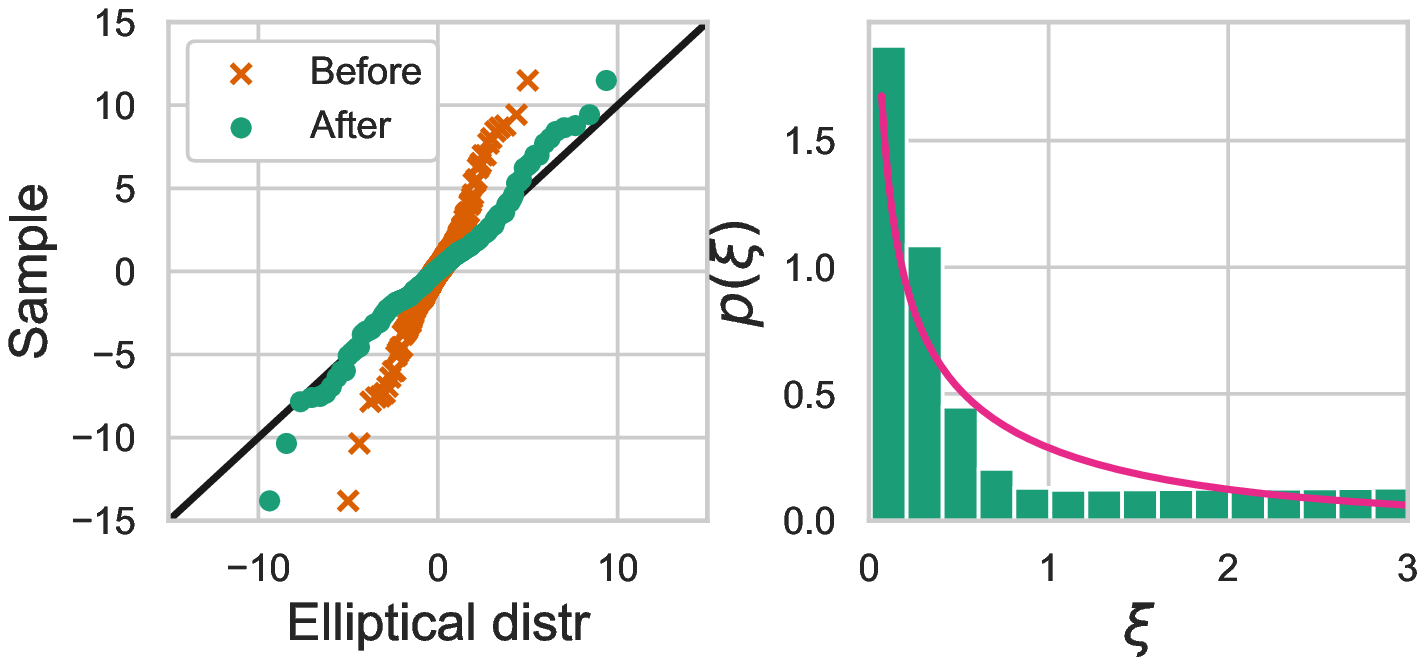}
\end{subfigure}

\begin{subfigure}{\columnwidth}
\caption{Truncated Laplace distribution}\label{fig:qq_plot_2}
  \includegraphics[width=0.9\linewidth]{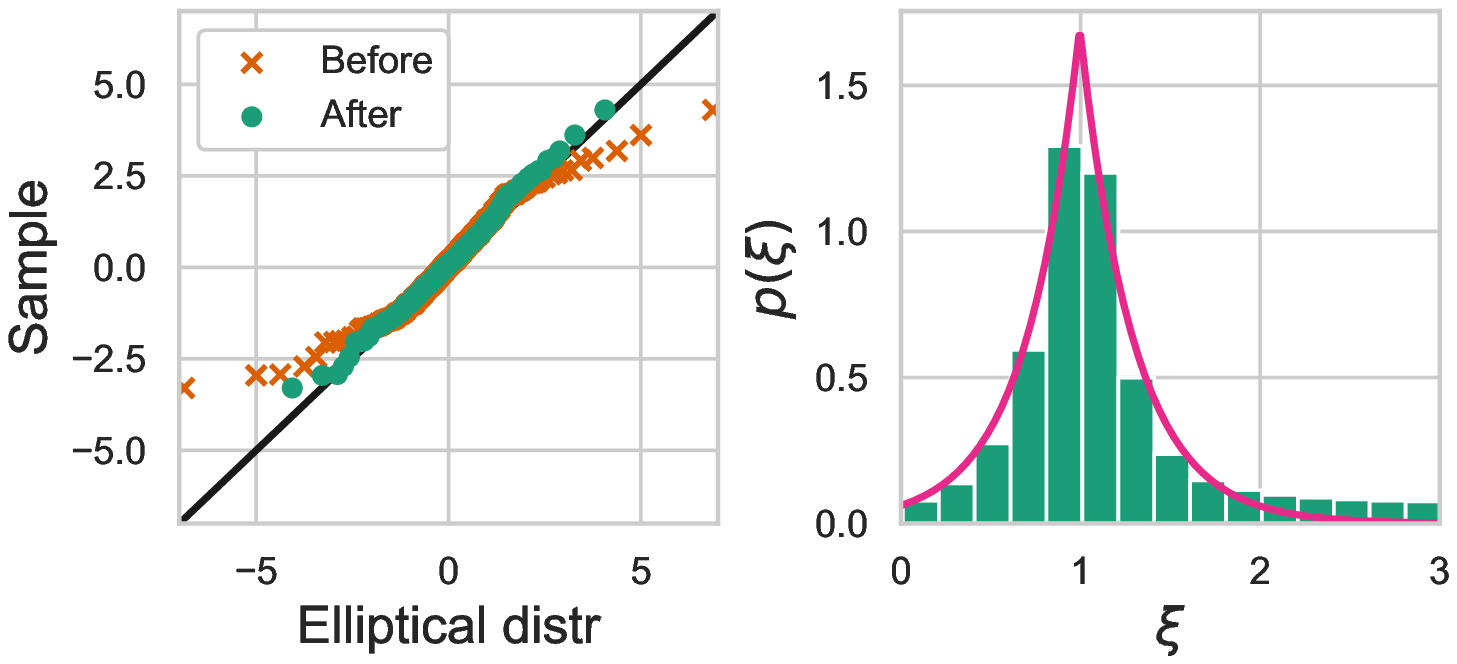}
\end{subfigure}
\vskip 0.2in
\caption{Tail behavior of the elliptical distribution with piecewise constant mixing distribution, when trained on synthetic data with known mixing distribution. The left column shows quantile-quantile plots of the model quantiles against the sample quantiles, before and after training. The right column shows the mixing distribution of the data generator and the piecewise constant mixing distribution after training.}
\label{fig:qq_plot}
\end{figure}
We wanted to investigate how well the elliptical distribution with piecewise constant mixing distribution captures fat-tailed behavior. We created two sets of training data by drawing 200 samples each from two elliptical distributions with different mixing distributions (see Figure \ref{fig:qq_plot}): a scaled chi-square distribution with scale parameter $\eta = 1$ and a Laplace distribution with its mode at $\xi=1$. Both distributions were truncated to be strictly larger than zero. 

We initialized the model with a uniform mixing distribution, i.e. all $h_i$ equal, and trained it by minimizing the Kullback-Leibler divergence between the training data and the model. Figure \ref{fig:qq_plot} shows the mixing distribution after training. Clearly, the training aligns the model quantiles to the sample quantiles on a relatively straight line, which indicates that the tail behavior of both data sets was captured reasonably well.


\begin{figure*}[t!]
    \vskip 0.05in
    \centerline{\includegraphics[width=1.7\columnwidth]{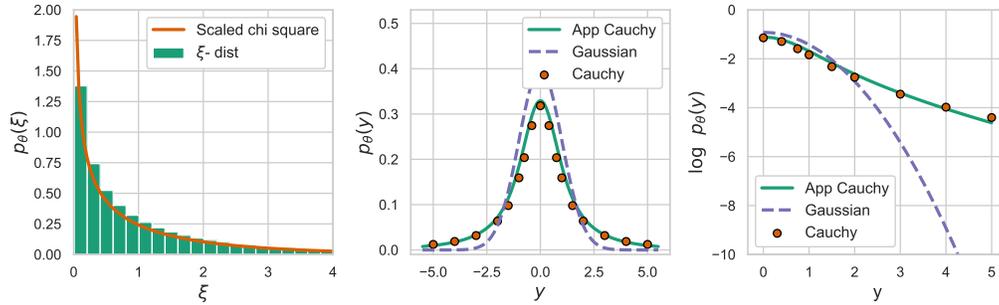}}
    \vskip 0.2in
    \caption{\textbf{Left}: approximation of the scaled chi-square distribution with a piecewise constant distribution. \textbf{Center:} comparison of the elliptical distribution that approximates the Cauchy distribution to the true Cauchy distribution. \textbf{Right:} comparison, in log-scale, of the tails of the elliptical distribution, the true Cauchy distribution and the Gaussian distribution. We can see that the elliptical distribution is close to (but not identical to) the true Cauchy distribution.}
    \label{fig:Cauchy_pdf}
\end{figure*}

\subsection{The Approximated Cauchy Process}

In this section we used the elliptical process with the parameters $h_1, \hdots, h_M$ fixed. This may be useful if we have some a priori knowledge of the noise characteristics or if we seek a specific behavior from the process. For example, to get a robust stochastic process we can approximate a Cauchy process with the elliptical process.

\begin{figure}[htbp]
    \vskip 0.1in
    \centerline{\includegraphics[width=0.95\columnwidth]{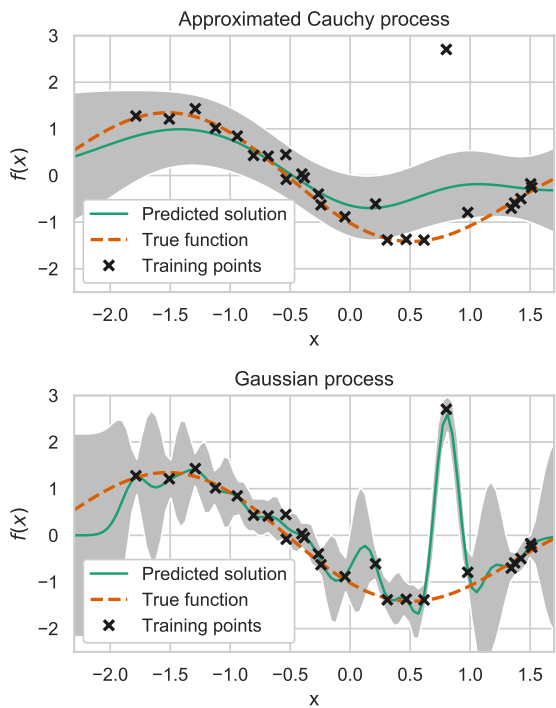}}
    \vskip 0.1in
    \caption{Example of the approximated Cauchy process for data, corrupted with outliers. \textbf{Top:} optimized solution for the approximate Cauchy process. \textbf{Bottom:} optimized solution for the Gaussian process. Both models were trained by minimizing the log-likelihood. The Gaussian process is more prone to overfit to the outliers while the Cauchy process perceive them as noise. The grey area is the 95 \% confidence interval.}
    \label{fig:Cauchy_toy}
\end{figure}

The Cauchy distribution is in a sense a pathological distribution, since both its expected value and covariance are undefined. Nevertheless we can approximate it, as shown in Figure \ref{fig:Cauchy_pdf}, by letting the mixing distribution approximate a scaled chi-square distribution with scale parameter $\eta = 1$. To get a valid expected value and covariance we let $\ell_0$ be a tiny positive number.

Figure \ref{fig:Cauchy_toy} illustrates how the approximated Cauchy process behaved when the data were corrupted with outliers (Cauchy noise). The fat tail made it tolerate outliers far away from the predicted expected value and, in contrast to the Gaussian process, the Cauchy process found a solution close to the true function. 

We compared the approximate Cauchy process with the elliptical process when $h_i$ were free and the Gaussian process on three different synthetic data sets (explained in the next section).

We see in Table \ref{table:results_synthetic} that the Cauchy process performed best on average out of the models when the training data were corrupted with Cauchy noise. When $\eta = 8$, all three models had comparable results and the benefit of using the approximated Cauchy process or the elliptical process had diminished. 

\begin{table*}[t!]
    \small
    \caption{The predictive Mean Square Error (MSE) and the test log-likelihood (LL) from the experiments with the synthetic data sets. We show the mean and the standard deviation from the 100 different runs.}
    \label{table:results_synthetic}
    \begin{center}
        \begin{tabular}{lllllll}
            \multicolumn{1}{c}{\bf }  
            &\multicolumn{2}{c}{\bf GP} &\multicolumn{2}{c}{\bf EP} & \multicolumn{2}{c}{\bf CaP}\\
            \multicolumn{1}{c}{\bf DATA SET}  
            &\multicolumn{1}{c}{ MSE}
            &\multicolumn{1}{c}{ LL}
            &\multicolumn{1}{c}{ MSE}
            &\multicolumn{1}{c}{ LL}
            &\multicolumn{1}{c}{ MSE}
            &\multicolumn{1}{c}{ LL}\\
            \hline \\
            Synth $\eta = 1$    & $0.29 \pm 0.64 $ & $-0.90\pm 0.11 $ &$ {0.068 \pm 0.14} $ & $-0.68 \pm 0.061 $  & ${0.068\pm 0.14}$ & ${-0.52 \pm 0.076}$\\
            Synth $\eta = 3$    &$ 0.15 \pm 0.16 $ & $-0.92 \pm 0.079 $  & ${0.13 \pm 0.09 }$ & $-0.73\pm 0.12 $ & ${0.13 \pm 0.09}$ & ${-0.59 \pm 0.14}$ \\
            Synth $\eta = 8$    &$ 0.13 \pm 0.09 $ & $-0.88 \pm 0.11 $  & ${0.13 \pm 0.10} $ & $-0.69\pm 0.15$ & $0.13\pm 0.09 $ & ${-0.56 \pm 0.17}$ \\
        \end{tabular}
    \end{center}
\end{table*}

\subsection{Regression On Synthetic And Real World Data Sets}
We tested the elliptical process on synthetic and real-world data sets. We especially wanted to test the model on data sets with outliers. The six data sets were:

\textbf{Synthetic data set:} We drew 100 samples from a Gaussian process prior, where 50 samples are used for training and 50 samples for testing. We added Student-t noise to the training data with three different values on the scale parameter $\eta = 1, 3, 10$. We trained the elliptical process with free $h_i$ and on the approximated Cauchy process explained in the previous section. 

\textbf{The Neal data set}, is a synthetic data set first proposed by \citet{Neal1997} as a data set containing target outliers. The target value was corrupted with Gaussian noise with standard deviation 0.1. 5 \% of the targets were replaced by outliers sampled form a Gaussian distribution. We generated 100 training samples and tested the model on 100 noise free samples.

\textbf{The Friedman data set}, proposed by \citet{friedman1991multivariate} is a synthetic data set derived from a non-linear function with 5 inputs. On top of that, 5 unrelated inputs are added which makes the input features 10-dimensional. We generated 100 training samples and tested the model on 100 noise free samples.


\textbf{Spatial Interpolation data set} was originally published by \citet{Dubois}. The data set consists of 467 daily rainfall measurements made in Switzerland in May 1986. 100 observed locations were used for training and the remaining 367 locations were used for testing. 

\textbf{Boston housing data set} was originally published by \citet{harrison1978hedonic}. There are 506 samples and 13 feature variables in this dataset. The targets are prices on houses in the Boston area. We randomly sampled 360 training points and 40 testing points.

\textbf{The Concrete data set} \citep{yeh1998modeling} has 8 input variables and 1030 observations. The target variables are the concrete compressive strength. We randomly selected 360 samples as training data and 40 samples as testing data. 

We iterated the experiments 100 times. The mean squared error (MSE) and the test log-likelihood (LL) for the synthetic data sets are reported in Table \ref{table:results_synthetic}. The results for the remaining data sets are reported in Table \ref{table:results}. The synthetic data sets results are also presented in Appendix \ref{appendix:results}. 

\begin{table*}[ht]
    \small
    \caption{The predictive Mean Square Error (MSE) and the test log-likelihood (LL) from the experiments. We show the mean and the standard deviation from the 100 different runs.}
    \label{table:results}
    \begin{center}
    \begin{tabular}{lllll}
        \multicolumn{1}{c}{\bf }  
        &\multicolumn{2}{c}{\bf GP} &\multicolumn{2}{c}{\bf EP}\\
        \multicolumn{1}{c}{\bf DATA SET}  
        &\multicolumn{1}{c}{ MSE}
        &\multicolumn{1}{c}{ LL}
        &\multicolumn{1}{c}{ MSE}
        &\multicolumn{1}{c}{ LL}
        \\
        \hline \\
        Neal& $0.025\pm 0.048 $ & $0.008 \pm0.27$    & $ 0.022\pm0.023 $ & $0.25  \pm 0.27 $  \\
        Friedman                &$ 0.028 \pm 0.01 $ & $-0.017 \pm 0.081 $  & $0.028 \pm 0.01 $ & $0.097\pm 0.085 $ \\
        Spatial           & $0.97 \pm 0.15 $  & $-1.42\pm0.11 $     & $0.89\pm 0.13$    & $-1.35 \pm0.06 $   \\
        Boston                & $0.14 \pm 0.098 $& $-0.98 \pm 0.20 $ & $0.13 \pm 0.094 $  & $-0.97 \pm 0.18 $  \\
        Concrete                &$0.15 \pm 0.06$ & $-0.99\pm0.12 $  & $0.15 \pm 0.06$ & $-0.98 \pm 0.11 $  \\
    \end{tabular}
    \end{center}
\end{table*}

From the experiments we see that the elliptical process works just as well as the Gaussian process and sometimes better. The elliptical process gives better mean square error and log-likelihood when the training data is corrupted with outliers. This is according to our hypothesis: the Gaussian process, with the thin tail, overfits more often to the outliers while a heavier tail perceives the outliers as noise.

\section{DISCUSSION}
The Gaussian distribution is the default in statistical modeling, for good reasons. Even so, far from everything is Gaussian---casually pretending it is, comes at your own risk. The elliptical distribution is an attractive alternative, which offers an increased flexibility in modeling the tails. The same reasoning applies when comparing the Gaussian process against the elliptical process. An entirely feasible approach is thus to start from the weaker assumptions of the elliptical process and let the data decide whether Gaussianity is, in fact, supported by the evidence.

In this paper we exemplified this for robust regression, where the thin tail of the Gaussian distribution makes it sensitive to outliers. But, applications that rely on accurate tail modeling are abound: anomaly detection \citep{Chandola2009}, Bayesian optimization \citep{Snoek2012}, optimal design \citep{Sjolund2017}, robust control \citep{Calafiore2006} and  modeling extreme values or rare events \citep{deHaan2007} to mention a few.


The elliptical process is applicable whenever the Gaussian process is applicable, including generative modeling \citep{Casale2018}. It is natural to sample the elliptical process by first sampling the mixing variable $\xi$ and then sampling $y$ conditioned on $\xi$. Since the latter follows a Gaussian distribution, this sampling procedure is efficient if sampling $\xi$ is easy (which it is for the piecewise constant distribution).

The ease of sampling also hints at the possibility of training the model using sampling based approximate inference instead of empirical Bayes. In our experiments we observed that the training often converged to globally sub-optimal solutions, as indicated by the fact that the Gaussian process---which is a subset of the elliptical process---occasionally returned better results. Putting prior distributions over the hyperparameters would serve as a practical form of regularization. 


The piecewise constant mixing distribution is convenient to work with and is straightforward to extend to any piecewise constant distribution. But other distributions that give closed-form expressions exist, and it is conceivable that some of them would be better suited for optimization. To mention a few options, the constant pieces could be replaced by piecewise polynomials, or the entire mixing distribution could be chosen as a Gamma distribution. 

\section{CONCLUSIONS}
We have presented a general construction of elliptical processes, based on Kano's theorem classifying the consistent elliptical distributions as those that can be represented as a mixture of Gaussian distributions. We exemplified the elliptical process with a piecewise constant mixing distribution, and illustrated some of the flexibility it offers, e.g. the approximated Cauchy process. We performed experiments on robust regression, where we compared the elliptical process with the Gaussian process, and found that, as expected, the elliptical process was more accurate in the presence of outliers.

The added flexibility of the elliptical processes could benefit a range of applications, both classical and new. We curiously look forward to what the future holds.
\subsubsection*{Acknowledgements}

\bibliographystyle{plainnat}
\bibliography{refs}
\onecolumn
\begin{appendix}
\title{The Elliptical Processes: \\
Supplementary Materials}
\maketitle

\section{THE ELLIPTICAL DISTRIBUTION}
From Theorem \ref{theorem:consistent} we know that a consistent elliptical distribution must mus originate from the integral
\begin{equation}\label{eq:consistancy_appendix}
    p_{\boldsymbol{\theta}}(u)  =|\boldsymbol{\Sigma}|^{-\frac{1}{2}} \int_0^{\infty} \left ( \frac{\xi}{2\pi}\right)^{\frac{n}{2}}e^{\frac{-u\xi}{2}}p_{\boldsymbol{\theta}}(\xi)d\xi.
\end{equation}
By definition, $p_{\boldsymbol{\theta}}(u)$ is a scaled mixture of normal distributions with the stochastic representation
\begin{equation}
    \mathbf{Y}| \ \xi \sim \mathcal{N}(\boldsymbol{\mu, \boldsymbol{\Sigma}}\frac{1}{\xi}), \ \ \ \xi \sim p_{\boldsymbol{\theta}}(\xi).
\end{equation}
By using the following representation of the elliptical distribution,
\begin{equation}
    \mathbf{Y} = \boldsymbol{\mu} + \boldsymbol{\Sigma}^{1/2}\mathbf{Z}\xi^{-1/2},
\end{equation}
where $\mathbf{Z}$ follows the standard normal distribution, we get the mean 
\begin{equation}
    \E[\mathbf{Y}] = \boldsymbol{\mu} + {\boldsymbol{\Sigma}}^{1/2}\E \left [\mathbf{Z}\right ] \ \E[\xi^{-1/2}] = \boldsymbol{\mu}
\end{equation}
and the covariance
\begin{align}
    \Cov(\mathbf{Y}) &= \E \left [(\mathbf{Y}- \boldsymbol{\mu})(\mathbf{Y}- \boldsymbol{\mu})^T \right ] \nonumber \\
    &= \E \left [(\boldsymbol{\Sigma}^{1/2}\mathbf{Z} / \sqrt{\xi})(\boldsymbol{\Sigma}^{1/2}\mathbf{Z} / \sqrt{\xi})^T \right ] \nonumber\\
    &=\E \left [{\xi}^{-1}\boldsymbol{\Sigma}^{1/2}\mathbf{Z}\mathbf{Z}^T(\boldsymbol{\Sigma}^{1/2})^T  \right ]\nonumber \\
    &=\E\left [ {\xi}^{-1}\right] \boldsymbol{\Sigma}.
\end{align}
The variance is a scale factor of the scale matrix $\boldsymbol{\Sigma}$. To get the variance we have to derive $\E\left [ {\xi}^{-1}\right]$. Note that if $\xi$ follows the scale chi-square distribution, $E[\xi^{-1}] = \nu/(\nu-2)$. We recognize form the Student-t distribution, where $\Cov(\mathbf{Y}) = \nu/(\nu-2)\mathbf{\Sigma}$.

\section{CONDITIONAL DISTRIBUTION}\label{appendix:conditional_dist}
\subsection{Proof of Proposition \ref{prop:conditional_distribution}}
\label{appendix:proposition_conditional}

To prove Proposition \ref{prop:conditional_distribution}, 
we partition the data $\mathbf{y}$ as $[\mathbf{y}_1, \mathbf{y}_2]$, so $n_1$ data points belong to $\mathbf{y}_1$, $n_2$ data points belong to $\mathbf{y}_2$ and $n_1 + n_2 = n$. 

The write joint distribution of  $[\mathbf{y}_1, \mathbf{y}_2]$ as $p(\mathbf{y}_1,\mathbf{y}_2|\xi)p_{\boldsymbol{\theta}}(\xi)$. The conditional distribution of $\mathbf{y}_2$, given $\mathbf{y}_1$ is then $p(\mathbf{y}_2|\mathbf{y}_1,\xi)p_{\boldsymbol{\theta}}(\xi|\mathbf{y}_1)$.

For a given $\xi$, $p(\mathbf{y}_2|\mathbf{y}_1,\xi)$ is the conditional normal distribution and so
\begin{equation}
   p(\mathbf{y}_2|\mathbf{y}_1,\xi)\sim \mathcal{N}(\boldsymbol{\mu}_{2|1}, \Sigma_{22|1}\hat{\xi}^{-1}), \ \ \ \hat{\xi} \sim p_{\boldsymbol{\theta}}(\xi|\mathbf{y}_1)
\end{equation}
where,
\begin{align}
    \boldsymbol{\mu}_{2|1} &= \boldsymbol{\mu}_2 + \Sigma_{21}\Sigma_{11}^{-1}(\mathbf{X}_1-\boldsymbol{\mu}_1) \\
     \Sigma_{22|1} &=  \Sigma_{22}- 
     \Sigma_{21}\Sigma_{11}^{-1} \Sigma_{21},
\end{align}
the same as for the conditional Gaussian distribution.
We obtain the conditional distribution $p_{\boldsymbol{\theta}}(\xi | \mathbf{y}_1)$ by first remembering that 
\begin{equation}
    p(\mathbf{y}_1| \xi) \sim \mathcal{N}(\boldsymbol{\mu}_{1}, \Sigma_{11}\frac{1}{\xi}).
\end{equation}  
Using Bayes’ Theorem we get:
\begin{align}\label{eq:cond_xi}
    p_{\boldsymbol{\theta}}(\xi | \mathbf{y}_1) &\propto p(\mathbf{y}_1| \xi)p_{\boldsymbol{\theta}}(\xi)\nonumber \\
    & \propto \left |\Sigma_{11} \frac{1}{\xi}\right|^{-1/2}\exp \left \{  -\xi\frac{u_1}{2}  \right\} p_{\boldsymbol{\theta}}(\xi) \nonumber\\
    &\propto \xi^{n_1/2} \exp \left \{  -\xi\frac{u_1}{2}  \right\} p_{\boldsymbol{\theta}}(\xi).
\end{align}
Recall that $u_1 = (\mathbf{y}-\boldsymbol{\mu}_1)^T\boldsymbol{\Sigma}_{11}^{-1}(\mathbf{y}-\boldsymbol{\mu}_1)$). We normalize the distribution by
\begin{equation}\label{eq:normalize_xi}
    c_{n_1, \boldsymbol{\theta}}^{-1} = \int_0^{\infty}
    \xi^{n_1/2} \exp \left \{  -\xi\frac{u_1}{2}  \right\} p_{\boldsymbol{\theta}}(\xi) d\xi
\end{equation}
The conditional mxing distribution is 
\begin{equation}\label{eq:cond_mixing}
    p_{\boldsymbol{\theta}} (\xi|\mathbf{y}_1) =  c_{n_1, \boldsymbol{\theta}}\xi^{n_1/2} \exp \left \{  -\xi\frac{u_1}{2}  \right\} p_{\boldsymbol{\theta}}(\xi)
\end{equation}
The conditional distribution of $\mathbf{y}_2$, given $\mathbf{y}_1$, is derived by using the consistency formula \eqref{eq:consistancy_appendix}
\begin{equation}
    p(\mathbf{y}_2|\mathbf{y}_1) =
\frac{1}{|\boldsymbol{\Sigma}_{22|1}|^{1/2}(2\pi)^{n_2/2}} \int_0^{\infty}\xi^{n_2/2}e^{-u_{2|1}\xi/2}p(\xi|\mathbf{y}_1)d\xi,
\end{equation}
where $u_{2|1} = (\mathbf{y}_2-\boldsymbol{\mu}_{2|1})^T\Sigma_{22|1}^{-1}(\mathbf{y}_2-\boldsymbol{\mu}_{2|1})$.
Using Equation \eqref{eq:cond_mixing} we get
\begin{equation}\label{eq:prop_condiitonal}
   p(\mathbf{y}_2|\mathbf{y}_1)  =
\frac{c_{n_1, \boldsymbol{\theta}}}{|\boldsymbol{\Sigma}_{22|1}|^{1/2}(2\pi)^{n_2/2}} \int_0^{\infty}\xi^{n/2}e^{-(u_{2|1} + u_1)\xi/2}p_{\boldsymbol{\theta}}(\xi)d\xi
\end{equation}

\section{ADD NOISE TO THE ELLIPTICAL PROCESS FOR REGRESSION}\label{appendix:added_noise}
If $\mathbf{F}$ is the stochastic variable of the latent function of the elliptical process and $\mathbf{Z}$ is the zero mean stochastic variable of the noise. Given $\xi$, $\mathbf{F}$ and $\mathbf{Z}$ are Gaussian, thus uncorrelated and independent. We use the law of total expectation to show that they are uncorrelated when $\xi$ is not given.
\begin{align}
    \E \left [  \mathbf{F}\,\mathbf{Z} \right ]&= \E_\xi \left [\E \left [ \mathbf{F} \, \, \mathbf{Z} \right| \xi]\right ]  \\ &=  \E_\xi \left [\E\left [ \mathbf{Z} | \xi \right] \ \ \E \left [ \mathbf{F} |\xi \right]\right ] \nonumber \\
\end{align}
but $\E\left [ \mathbf{Z} | \xi \right] = 0$ so we get
\begin{equation}
    \E \left [  \mathbf{F}\,\mathbf{Z} \right ] = 0.
\end{equation}
which is the same as saying than $\mathbf{F}$ and $\mathbf{Z}$ are uncorrelated. If $\mathbf{F}$ and $\mathbf{Z}$ where independent, $ \E \left [  \mathbf{F}\,\mathbf{Z} \right ] = \E  [  \mathbf{F} ]\ \E  [  \mathbf{Z} ]$, which is not true since both $\mathbf{F}$ and $\mathbf{Z}$ depends on $xi$.

\section{EXAMPLE: PIECEWISE CONSTANT}\label{appendix:deriving_sq_box_dis}
We start with deriving the piecewise constant distribution with one piece, so $p_{\boldsymbol{\theta}}(\xi) =h_1\mathbbm{1}_{(\ell_{0}, \ell_0 +\Delta)}(\xi) $, a constant function between $\ell_0$ and $\ell_0 +\Delta$. A probability distribution must normalize to one, so $h_1 = 1/\Delta$. The elliptical distribution is 
\begin{equation}
        p_{\boldsymbol{\theta}}(u)  =|\boldsymbol{\Sigma}|^{-\frac{1}{2}} \int \left ( \frac{\xi}{2\pi}\right)^{\frac{n}{2}}e^{\frac{-u\xi}{2}}\frac{1}{\Delta}\mathbbm{1}_{(\ell_{0}, \ell_0 +\Delta)}(\xi)  d\xi =
        \frac{|\boldsymbol{\Sigma}|^{-\frac{1}{2}}}{\Delta(2 \pi)^{n/2}} \int_{l_0}^{l_0 + \Delta}  {\xi}^{\frac{n}{2}}e^{\frac{-u\xi}{2}} d\xi
\end{equation}
This is, by using the variable substitution $t = u\xi/2$, the incomplete gamma function 
\begin{equation}\label{eq:incomplete_gamma}
    \Gamma(s, a, b) = \int_a^bt^{s-1}e^{-t}dt.
\end{equation}
The elliptical distribution with one piece is
\begin{equation}\label{eq:derived_phi}
     p_{\boldsymbol{\theta}}(u) =\frac{2u^{-(n/2+1)}}{\pi^{\frac{n}{2}}|\boldsymbol{\Sigma}|^{\frac{1}{2}}} \frac{1}{\Delta}\Gamma \left (\frac{n}{2}+1,u \frac{\ell_0}{2}, u\frac{(\ell_0 +\Delta)}{2}\right ).
\end{equation}
We extend the distribution to $M$ pieces with heights $h_1, h_2, \hdots, h_M$ and width $\Delta$. It total area of the distribution is $\Delta \sum_k h_k$ (the normalization constant) so the $\xi$-distribution is
\begin{equation}
    p_\theta(\xi) =  \frac{1}{\Delta \sum_k h_k}\sum_{k=1}^M h_i \mathbbm{1}_{(\ell_{k-1}, \ell_k)},
\end{equation}
where 
\begin{equation}
    \ell_k = \ell_0 + k\Delta.
\end{equation}
The Elliptical distribution derives to
\begin{equation}\label{eq:derived_phiM}
     p_{\boldsymbol{\theta}}(u) =\frac{2u^{-(n/2+1)}}{\pi^{\frac{n}{2}}|\boldsymbol{\Sigma}|^{\frac{1}{2}} \Delta \sum_k h_k}\sum_{k=1}^M h_k\Gamma\left ( \frac{n}{2}+1,\frac{u \ell_{ k-1}}{2}, \frac{u\ell_k}{2} \right )
\end{equation}
The derivatives of the incomplete Gamma are:
\begin{equation}
    \frac{d}{da} \Gamma(s, a, b) = -a^{s-1}e^{-a},
\end{equation}
\begin{equation}
    \frac{d}{db} \Gamma(s, a, b) = b^{s-1}e^{-b},
\end{equation}
We use these derivatives when implementing the incomplete gamma function with back-propagation in PyTorch. 

\subsection{Conditional Distribution}
We derive the conditional distribution by using Equation \eqref{eq:prop_condiitonal}.
\begin{align*}
   p(\mathbf{y}_2|\mathbf{y}_1)  &=
\frac{c_{n_1, \boldsymbol{\theta}}}{|\boldsymbol{\Sigma}_{22|1}|^{1/2}(2\pi)^{n_2/2}} \int_0^{\infty}\xi^{n/2}e^{-(u_{2|1} + u_1)\xi/2}p_{\boldsymbol{\theta}}(\xi)d\xi \\
 &= \frac{c_{n_1, \boldsymbol{\theta}} (u_{2|1} + u_1)^{-(n/2+1)}}{|\boldsymbol{\Sigma}_{22|1}|^{1/2}\pi^{n_2/2}2^{-(n_1/2 +1)}\Delta \sum_i h_i}
    \sum_{k=0}^M h_k\Gamma \left ( \frac{n}{2}+1, \frac{\ell_{k-1}(u_{2|1} + u_1)}{2}, \frac{\ell_k(u_{2|1} + u_1)}{2} \right ).
\end{align*}
And we derive the normalization constant $c_{n_1, \boldsymbol{\theta}}$ from Equation \eqref{eq:normalize_xi}.
\begin{align*}
    c_{n_1, \boldsymbol{\theta}}^{-1} &= \int_0^{\infty}
    \xi^{n_1/2} \exp \left \{  -\xi\frac{u_1}{2}  \right\} p_{\boldsymbol{\theta}}(\xi) d\xi \\
     &= \left( \frac{u_1}{2}\right)^{-(n_1/2+1)}\frac{1}{\Delta \sum_i h_1}\sum_{k=1}^M h_k 
      \Gamma \left (  \frac{n_1}{2} +1, \frac{\ell_{k-1}u_1}{2}, \frac{\ell_ku_1}{2} \right ).
\end{align*}
The conditional probability distribution $p(\xi|\mathbf{y}_1)$ is
\begin{equation}
    p(\xi|\mathbf{y}_1) = c_{n_1, \boldsymbol{\theta}} 
    \xi^{n_1/2} e^{  -\xi\frac{u_1}{2}} p_{\boldsymbol{\theta}}(\xi)
\end{equation}
We get the conditional covariance by deriving the expectation value of $\xi^{-1}$
This expression can be rewritten in the same way as the normalization constant to, so
\begin{align}
     \E_{\xi|\mathbf{Y}_1}[\xi^{-1}]  &= \int_0^{\infty} \xi^{-1}p(\xi|\mathbf{y}_1) d\xi\\
     &=c_{n_1, \boldsymbol{\theta}}\int_0^{\infty} \xi^{-1}\xi^{n_1/2} e^{  -\xi\frac{u_1}{2}} p_{\boldsymbol{\theta}}(\xi) d\xi\\
     & = c_{n_1, \boldsymbol{\theta}} \left(\frac{u_1}{2}\right)^{-n_1/2}\frac{1}{\Delta \sum_i h_i }\sum_{k=1}^M 
   h_k\Gamma \left (\frac{n_1}{2}, \frac{u_1\ell_{k-1}}{2}, \frac{u_1\ell_k}{2} \right ) \\
    &= \frac{u_1}{2}\frac{ \sum_{k=1}^M h_k\Gamma \left (\frac{n_1}{2}, \frac{\ell_{k-1} u_1}{2},\frac{\ell_k u_1}{2} \right )
    }{
    \sum_{k=1}^M h_k \Gamma \left (\frac{n_1}{2}+1,\frac{\ell_{k-1} u_1}{2},\frac{\ell_k u_1}{2} \right)}. \nonumber
\end{align}
The conditional covariance of $\mathbf{y}_2$ given $\mathbf{y}_1$ is
\begin{equation}\label{eq:cond_variance}
    \Cov(\mathbf{y}_2| \mathbf{y}_1)= \frac{u_1}{2}\frac{ \sum_{k=1}^M h_k\Gamma \left (\frac{n_1}{2}, \frac{\ell_{k-1} u_1}{2},\frac{\ell_k u_1}{2} \right )
    }{
    \sum_{k=1}^M h_k \Gamma \left (\frac{n_1}{2}+1,\frac{\ell_{k-1} u_1}{2},\frac{\ell_k u_1}{2} \right)} \boldsymbol{\Sigma}_{22|1}.
\end{equation}

\section{IMPLEMENTATION}\label{appendix:implementation}
The model was implemented in PyTorch \citep{paszke2019pytorch}.
The incomplete gamma function is not part of the PyTorch library, so we implement it. A naive implementation sometimes results in numerical underflow, especially when $u$ is large. We solved this problem by standardizing the input data, using the package mpmath \citep{mpmath}, which can compute the incomplete gamma function with arbitrary precision, and scaling the incomplete gamma function. 

The incomplete gamma function grows exponentially and can easily cause numerical overflow. To overcome this issue, we use a scaled version of the incomplete gamma function,
\begin{equation}
    \Psi_i(s, \ell_{i-1}, \ell_i, u) = 
    e^{u/2} (u/2)^{-s} \Gamma\left ( s, \frac{\ell_{i-1} u}{2}, \frac{\ell_i u}{2} \right ). 
\end{equation}
We get
\begin{equation}\label{eq:derived_phiM_scaled}
     p_{\boldsymbol{\theta}}(u) =\frac{ e^{-u/2}} {(2\pi)^{\frac{n}{2}}|\boldsymbol{\Sigma}|^{\frac{1}{2}} \Delta \sum_i h_i}\sum_{k=1}^M h_k\Psi_k\left ( \frac{n}{2}+1, \ell_{ k-1}, \ell_k, u\right ).
\end{equation}
For the implementation we must differentiate $\Psi_i$ with respect to $u$
\begin{equation}
    \frac{\partial }{\partial u} \Psi_i(s, \ell_{i-1}, \ell_1, u)
    = \Gamma\left ( s, \frac{\ell_{i-1} u}{2}, \frac{\ell_i u}{2} \right ) 
    \frac{\partial}{\partial u}\left (    e^{u/2}(u/2)^{-s} \right ) 
    + e^{u/2}(u/2)^{-s}\frac{\partial}{ \partial u} \Gamma\left ( s, \frac{\ell_{i-1} u}{2}, \frac{\ell_i u}{2} \right )
\end{equation}
\begin{equation}
\frac{\partial}{\partial u} \left [ e^{u/2} \left (u/2\right)^{-s}\right ]= e^{u/2}\left( -\frac{s}{2}(u/2)^{-(s+1)} \right) +
 \frac{e^{u/2}(u/2)^{-s}}{2}
= \frac{e^{u/2}(u/2)^{-s}}{2}\left (
 1-\frac{s2}{u}
\right ) 
\end{equation}
\begin{equation}
    \frac{\partial}{ \partial u} \Gamma\left ( s, \frac{\ell_{i-1} u}{2}, \frac{\ell_i u}{2} \right ) = 
    u^{s-1}\left [- \left ( \frac{\ell_{i-1}}{2}\right )^{s}e^{-u\ell_{i-1}/2} + 
    \left ( \frac{\ell_{i}}{2}\right )^{s}e^{-u \ell_{i}/2}\right ] 
\end{equation}
and so
\begin{align}
    \frac{\partial }{\partial u} \Psi_i(s, \ell_{i-1}, \ell_i, u) &=
    \Gamma\left ( s, \frac{\ell_{i-1} u}{2}, \frac{\ell_i u}{2} \right )
    \frac{e^{u/2}(u/2)^{-s}}{2}\left (
 1-\frac{2s}{u}
\right )  + \frac{- \ell_{i-1}^{s}e^{-\frac{u}{2}(\ell_{i-1}-1)} + 
    \ell_{i}^{s}e^{-\frac{u}{2}(\ell_{i}-1)}}{u} \\ \nonumber
    &=\frac{\Psi_i(s, \ell_{i-1}, \ell_i, u)}{2}\left (
 1-\frac{2s}{u}
\right ) +  \frac{
    \ell_{i}^{s}e^{\frac{u}{2}(1-\ell_{i})
    }- \ell_{i-1}^{s}e^{\frac{u}{2}(1-\ell_{i-1})} }{u}
\end{align}
For the conditional variance $\Cov(\mathbf{y}_2| \mathbf{y}_1)$, derived in \eqref{eq:cond_variance}, we get
\begin{align}
    \Cov(\mathbf{y}_2| \mathbf{y}_1)=& \frac{u_1}{2}\frac{ \sum_{k=1}^M h_k\Gamma \left (\frac{n_1}{2}, \frac{\ell_{k-1} u_1}{2},\frac{\ell_k u_1}{2} \right )
    }{
    \sum_{k=1}^M h_k \Gamma \left (\frac{n_1}{2}+1,\frac{\ell_{k-1} u_1}{2},\frac{\ell_k u_1}{2} \right)} \boldsymbol{\Sigma}_{22|1} \nonumber \\
    =&\frac{e^{-u/2} (u/2)^{n/2}}{e^{-u/2} (u/2)^{n/2+1}} \frac{u_1}{2}\frac{ \sum_{k=1}^M h_k\Psi_k\left ( \frac{n}{2}, \ell_{ k-1}, \ell_k, u\right )
    }{
    \sum_{k=1}^M h_k \Psi_k\left ( \frac{n}{2}+1, \ell_{ k-1}, \ell_k, u\right )} \boldsymbol{\Sigma}_{22|1} \nonumber \\
    =& 
    \frac{ \sum_{k=1}^M h_k\Psi_k\left ( \frac{n}{2}, \ell_{ k-1}, \ell_k, u\right )
    }{
    \sum_{k=1}^M h_k \Psi_k\left ( \frac{n}{2}+1, \ell_{ k-1}, \ell_k, u\right )} \boldsymbol{\Sigma}_{22|1} 
\end{align}
The negative log likelihood is
\begin{align}
    -\log p_{\boldsymbol{\theta}}({u}) &= u/2+ \frac{n}{2}\log (2\pi) + \frac{1}{2} \log |\Sigma| + \log\left ( \Delta\sum_k h_k\right ) \nonumber 
    -\log\left (\sum_{k=1}^M 
    h_k\Psi_k\right ).\nonumber \\
    \label{eq:log_lokelihood2}
\end{align}
During training we minimize the negative log likelihood along with a smoothness penalty on the mixing distribution, $\lambda\sum_{i=2}^M (h_i-h_{i-1})^2$. $\lambda$ is set by cross validation.

\section{RESULTS}\label{appendix:results}
\begin{figure*}[ht]
\vskip 0.3in
    \centering
    \begin{subfigure}[b]{0.4\textwidth}
        \caption{The test mean square error (MSE)}
        \includegraphics[width = \columnwidth]{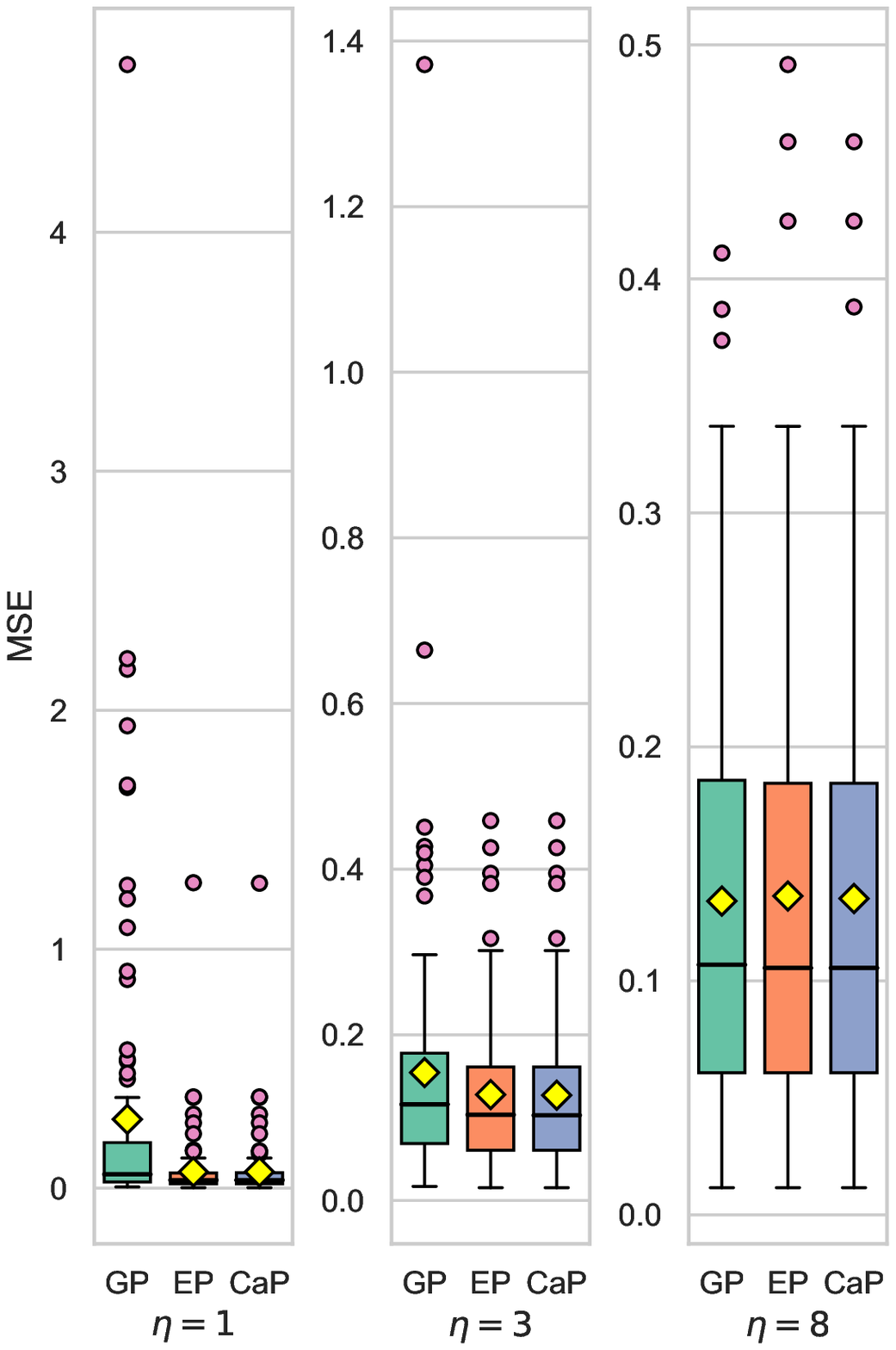}
    \end{subfigure}
    \begin{subfigure}[b]{0.45\textwidth}
        \centering
        \caption{The test log likelihood (LL)}\label{fig:synth_ll}
        \includegraphics[width = \columnwidth]{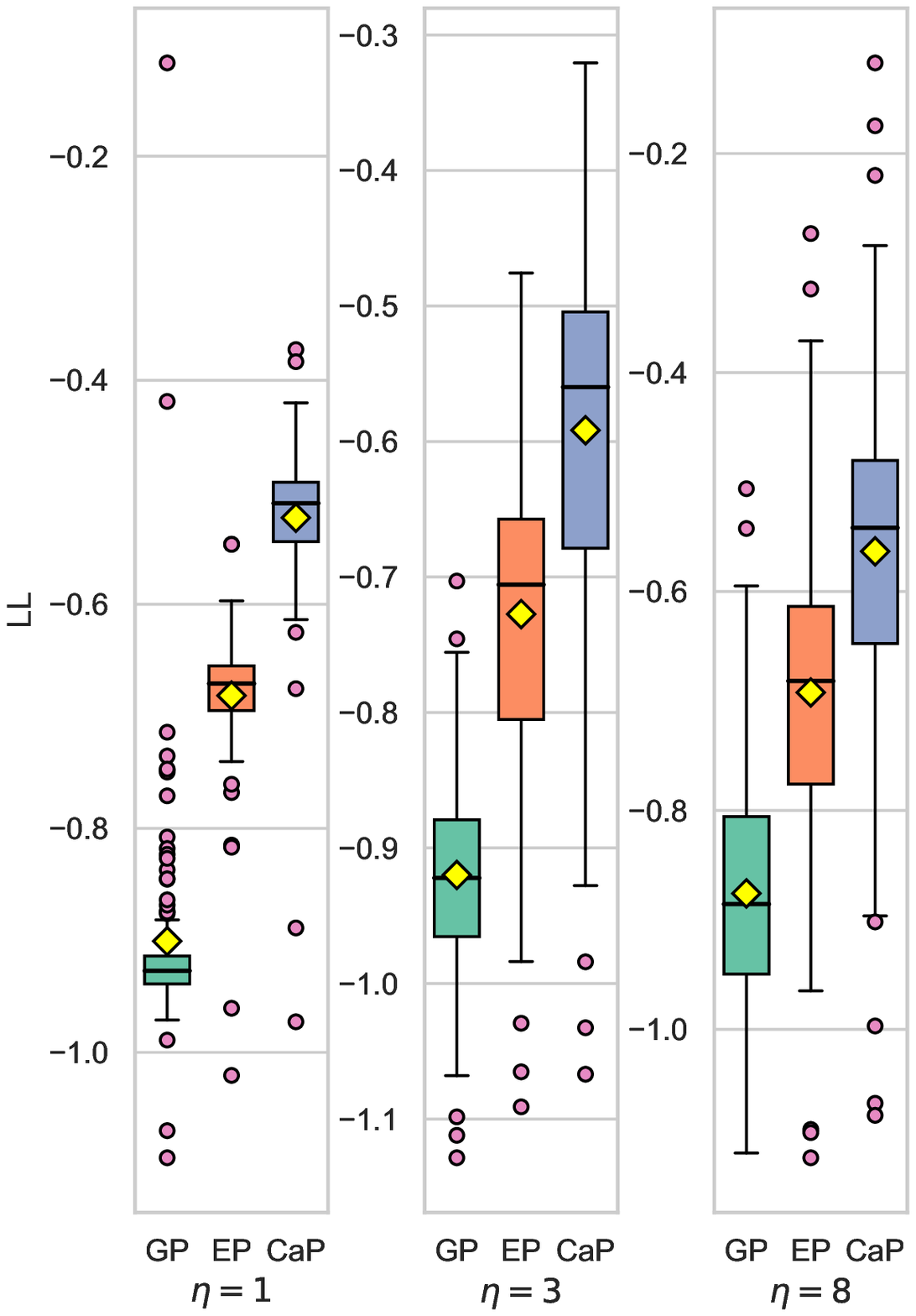}
    \end{subfigure}
    \vskip 0.3in
    \caption{The test mean square error (MSE) (\textbf{a}) and the test log likelihood (LL) (\textbf{b}) of the Synthetic data set. Student-t noise was added to the training data with scale parameter $\eta = 1, 3, 8$. The boxes show the first quartile to the third quartile of the values. The lines going out of the box show where of the smallest and the largest non-outlier value are. The yellow diamond is the mean, the black line is the median and the circular dots are outliers.}
    \label{fig:synth_datasets_mse_ll}
\end{figure*}
\clearpage
\end{appendix}

\end{document}